\documentclass[journal,12pt,onecolumn,draftclsnofoot,]{IEEEtran}

\ifCLASSINFOpdf
\else
   \usepackage[dvips]{graphicx}
\fi
\usepackage{url}
\usepackage{amsmath}
\usepackage{graphicx}
\usepackage{placeins}
\usepackage{float}
\usepackage{bbm}
\usepackage{subfigure}
\usepackage{verbatim}
\usepackage{amssymb}
\usepackage{amsmath}
\usepackage{amsthm}
\usepackage{mwe}
\usepackage{docmute}
\usepackage[ruled, vlined]{algorithm2e}

\usepackage[disable]{todonotes}

\usepackage{graphicx}

\newtheorem{definition}{Definition}
\newtheorem{lemma}{Lemma}

\theoremstyle{definition}
\newtheorem{remark}{Remark}
\newtheorem{assumption}{Assumption}

\newcommand{\oldbluetext}[1]{\ignorespaces}
\newcommand{\blue}[1]{\textcolor{black}{#1}}
\newcommand{\revision}[1]{\textcolor{black}{#1}}

\begin{document}

\title{Hybrid Cognition for Target Tracking in Cognitive Radar Networks}
\author{William W. Howard, R. Michael Buehrer
\thanks{W.W. Howard and R.M. Buehrer are with {\it Wireless @ Virginia Tech}, Bradley Department of ECE, Virginia Tech, Blacksburg, VA, 24061. \\ 
Correspondence:$\{${wwhoward}$\}$@vt.edu  \\
Portions of this work were presented at IEEE MILCOM 2022, Rockville, MD, December 2022 \cite{howard2022_decentralized_conf}. }
}

\maketitle
\pagenumbering{roman}

\begin{abstract}
This work investigates online learning techniques for a cognitive radar network utilizing feedback from a central coordinator. 
\blue{The available spectrum is divided into channels, and each radar node must transmit in one channel per time step. }
The network attempts to optimize radar tracking accuracy by learning the optimal channel selection for spectrum sharing and radar performance. 
We define optimal selection for such a network in relation to the radar observation quality obtainable in a given channel. 
\blue{This is a difficult problem since the network must seek the optimal assignment from nodes to channels, rather than just seek the best overall channel. }
Since the presence of primary users appears as interference, the approach also improves spectrum sharing performance. 
In other words, maximizing radar performance also minimizes interference to primary users. 
Each node is able to learn the quality of several available channels through repeated sensing. 
\revision{We define hybrid cognition as the condition where both the independent radar nodes as well as the central coordinator are modeled as cognitive agents, with restrictions on the amount of information that can be exchanged between the radars and the coordinator. }
Importantly, each part of the network acts as an online learner, observing the environment to inform future actions. 
We show that in interference-limited spectrum, where the signal-to-interference-plus-noise ratio \revision{varies by channel and over time for a target with fixed radar cross section}, a cognitive radar network is able to use information from the central coordinator in order to reduce the amount of time necessary to learn the optimal channel selection. 
We also show that even limited use of a central coordinator can eliminate collisions, which occur when two nodes select the same channel. 
We provide several reward functions which capture different aspects of the dynamic radar scenario and describe the online machine learning algorithms which are applicable to this structure. 
In addition, we study varying levels of feedback, where central coordinator update rates vary.
\oldbluetext{blue}{We compare our algorithms against baselines and demonstrate dramatic improvements in convergence time over the prior art. }
\blue{A network using hybrid cognition is able to use a minimal amount of feedback to achieve much faster convergence times and therefore lower tracking error. }
\end{abstract}


\section{Introduction}
This work seeks to improve the learning rate of a cognitive radar network (CRN) by introducing a central coordinator (CC) to provide limited feedback. 
Specifically in this work we address the role of a central coordinator within a cognitive radar network using an online learning strategy to achieve coordination as well as optimize radar tracking and spectrum sharing performance. 
Generally, radar networks achieve superior tracking performance than is possible for a single high-powered radar node \cite{554211}. 
This is due in part to the increased spatial diversity \cite{4102537} and spectral agility \cite{Martone_CRN_loop}. 
Distributed nodes can cover a greater area to perform detection, and can exploit more spatial degrees of freedom to more accurately estimate target parameters. 
However, to obtain this superior performance, the individual radar nodes which comprise the radar network must coordinate with each other to efficiently use the available spectrum and avoid causing harmful interference inside or outside the network. 
At the root, this problem is caused by a fundamental need to both explore the available channels and simultaneously exploit the best channels (in terms of tracking performance). 

Fixed, rule-based coordination has been proposed to solve this problem \cite{7272881} \cite{6892937}. 
This works well when the scenario parameters are well-known \emph{a priori}, but can suffer poorer performance when the scenario is more unpredictable (i.e., unknown targets or interference). 

Among other things, the desire for more flexible and adaptable systems motivated the initial research into cognitive radio and radar. 
Cognitive systems, at the core, are defined as possessing the ability to monitor the environment and modify operating parameters towards a goal \cite{haykin2006}. 
Further, Haykin provides the following dichotomy of cognitive networks: 
\begin{enumerate}
    \item \textbf{Distributed Cognition}, where observations from individual nodes are combined at a \emph{fusion center}\footnote{Fusion centers in this \blue{type of network} are assumed to perform no decision functions; i.e., they simply combine measurements and provide data to operators. } but no feedback is provided to the nodes. 
    \item \textbf{Centralized Cognition}, where a \emph{central coordinator}\footnote{Central coordinators are assumed to perform the functions of a fusion center \emph{as well as} performing some decision-making functions. } is the only cognitive agent, collecting observations from each node and dictating future actions. 
\end{enumerate}
Whether distributed or centralized, cognitive systems tend to be online learners due to the necessity to specialize to new, unknown environments and the difficulty of training the network ahead of time for an unknown environment.

Fully distributed cognition \cite{howard2022_MMABjournal} \cite{5535151} is useful when there is a desire for the parts of a CRN to be entirely disjoint and independent. 
Fully distributed approaches rely on consensus techniques \cite{8950073} \cite{mehrabian20a} to exchange information between the parts of a network and to determine optimal actions. 
In the radar context, these techniques can be very slow (requiring $10^4$ or greater time steps to converge to optimal actions) and can cause a large amount of mutual interference. 
Such a convergence rate can be problematic in some settings. 

Centralized cognition is not without trade-offs either. 
When cognition is limited to the CC, the individual nodes become over-reliant on the CC. 
The feedback costs can also grow immense, as we will show.

This work investigates \emph{hybrid cognition}, seeking the minimal amount of feedback necessary in a CRN to obtain near-optimal radar tracking performance in a short time without sacrificing node-level cognition. 
Our previous work \cite{howard2021_multiplayerconf} \cite{howard2022_MMABjournal} considered \emph{strictly} decentralized techniques, and did not assume the presence of a CC. 
While our approach was effective under these circumstances, the technique resulted in a relatively slow convergence rate. 
Our current work provides a generalization to model the CRN as containing a CC which can communicate and provide feedback to the radar nodes, with a goal of speeding up convergence. 

\subsection{Contributions} 
This paper makes the following contributions to the state of the art: 
\begin{itemize}
    \item The first work studying the role of feedback in cognitive radar networks.
    In particular, we study the case where cognition is divided between a Central Coordinator and the individual Cognitive Radar Nodes. 
    We do this by developing a framework for feedback, then structuring several algorithms which take advantage of different levels of feedback. 
    We show that there is a direct correlation between feedback and target tracking performance. 
    \item A system model for analyzing feedback in CRNs, where a CC provides data fusion as well as cognitive functions. 
    This is useful for future works, as such a model does not yet exist in the literature. 
    \item A mathematical analysis of the different reward functions available to learning algorithms in such a framework. 
    In addition, we discuss when approximations to these rewards may be merited. 
    \item We modify an existing decentralized algorithm \cite{mehrabian20a} to introduce feedback. 
    \item We supply simulations comparing our proposed model against techniques without feedback as well as an oracle which selects the actions which are best in hindsight. 
    \item We show that CRN performance can be significantly improved over short time horizons when feedback is used, and that even infrequent feedback is sufficient to improve convergence time in some scenarios. 
\end{itemize}

\subsection{Notation}
We use the following notation. 
Matrices and vectors are denoted as bold upper $\mathbf{X}$ or lower $\mathbf{x}$ case letters respectively.
Element-wise multiplication of two matrices or vectors is shown as $X\odot Y$. 
Functions are shown as plain letters $F$ or $f$. 
Sets $\mathcal{A}$ are shown as script letters. 
The cardinality $|\mathcal{A}|$ of a set $\mathcal{A}$ refers to the number of elements in that set. 
The transpose operation is $\mathbf{X}^T$. 
The set of all real numbers is $\mathbb{R}$ and the set of integers is $\mathbb{Z}$. 
The speed of electromagnetic radiation in a vacuum is given as $c$. 
The Euclidean norm of a vector $\mathbf{x}$ is written as $||\mathbf{x}||$. 
Estimates of a true parameter $p$ are given as $\hat{p}$. 
\revision{Bachmann-Landau asymptotic notation is written as $\mathcal{O}\left(\cdot\right)$. }

\subsection{Organization} 
The remainder of this paper is organized as follows. 
Section II discusses previous work in the field of cognitive radar networks and relevant machine learning. 
Section III provides the network system model assumed in this work. 
Section IV covers the relevant learning theory and details the reward models. 
Our proposed algorithms are discussed in Section V and Section VI provides simulations comparing our algorithms against several baselines. 
We draw conclusions in Section VII.



\section{Background}
\subsection{Related Previous Work}
\subsubsection{Cognitive Radar}
Cognitive radar (CR) has been the subject of intense study in recent years. 
In \cite{Martone_CRN_loop}, the authors survey recent work in spectrum sharing for cognitive radar. 
Since CR has inherent operational flexibility, it is \oldbluetext{blue}{natural} to implement spectrum sharing in environments where CR nodes are secondary users. 
\oldbluetext{blue}{Cognitive radar, as a field, has been investigated since the early 2000s \cite{8961364} \cite{8398580}. 
Various parameters have been exposed to cognitive decision-making: target parameter estimation, resource management, RF filtering, waveform selection, etc. }
\revision{Real-time implementation is investigated in \cite{9226487}, where the authors emply a sense-and-avoid strategy and a cognitive perception-action cycle. }
The authors of \cite{9455255} and \cite{thornton2022_universaljournal} investigate single-node cognitive radar and apply detailed machine learning techniques describing waveform selection techniques and adapting them to a broad class of target models.

Early research into cognitive systems was motivated in part by biological systems \cite{4141064}. 
Researchers wished to enable cognitive agents to display the adaptive intelligence \oldbluetext{blue}{and decision making capabilities} exhibited by biological systems. 
In general, this is accomplished through observation of the environment and use of statistical or machine learning algorithms to act on new information \cite{1593335}. 
\oldbluetext{blue}{Adaptive} radar systems are a \oldbluetext{blue}{good fit} for cognition \oldbluetext{blue}{since they can model} the echolocation abilities of bats \cite{simmons1973resolution}. 
\revision{The work of \cite{AGRAWAL2022101673} provides a review of cognitive processes applied to radar and radio networks. }

\subsubsection{Cogntive Radar Networks}
Cognitive radar networks have also been addressed in the literature, from their proposal in 2006 \cite{haykin2006} to more recent work. 
In general, the work on CRNs has been focused on time allocation (scheduling) or power allocation. 

In an early work on CRNs, the authors of \cite{6494389} propose a beamsteering strategy to split a search space between two radar nodes in a centralized CRN. 
This work showed a performance improvement in both detection and tracking over a network of two traditional radar nodes. 
While the problem addressed resource sharing in CRNs, it is limited in scope to two radar nodes. 

Several works focus on power and dwell time allocation in CRNs \cite{8974269}, \cite{6095653}, \cite{8970598}. 
These works consider CRNs sharing a single channel, which must allocate the limited observation time to the nodes of the networks. 
Instead of considering time division access schemes, our current work considers channelized spectrum. 
Further, many of these works consider pre-allocation schemes rather than the adaptive methods we consider here. 

Scheduling has been applied to the mutual interference problem in radar networks \cite{9545730}, with the goal of reducing pulse collisions within a CRN. 
While this method was shown to be effective and feasible, it relies on pre-allocation of resources which can fail to perform optimally in a dynamic environment where the mean \texttt{SINR} in each channel can vary in time. 

In addition, multistatic cognitive radar networks have been studied \cite{9455245}, where each radar in a network is able to receive and process the pulses transmitted by the other nodes. 
Multistatic radar operation allows for greater tracking accuracy, at a cost of greater amounts of coordination and processing.

\subsubsection{Machine Learning Applied to CRNs}

Statistical and machine learning (ML) approaches are natural for CRNs. 
Since cognitive nodes are able to observe the environment over time and choose from multiple actions, reinforcement learning is particularly well-suited. 
Reinforcement learning is the branch of machine learning that deals with sequential learning in possibly stochastic environments \cite{EXP3}. 
Since the exact interference and target behavior cannot be known in advance, approaches that adapt and generalize to broad classes of environments will out-perform those which depend on specific target and/or interference behavior.

As mentioned above, the purpose of this work is to investigate the balance between distributed and centralized learning. 
As such, we must primarily consider distributed learning models, and how they can be adapted in a hybrid framework. 
Distributed learning spreads components of a learning structure across nodes in a network \cite{shoham2008multiagent}. 
Obviously, distributed learning techniques come with several requirements. 
The selected algorithm must be well-suited to the environment. 
For example, the field of research into federated learning 
\cite{mcmahan2017communication} \cite{8950073}
investigates isolated models, trained on independent identically distributed (iid) data. 
This iid assumption is not valid in all environments; particularly, since all of the nodes in a CRN sample the same environment and are tracking the same targets, the observations are not independent.

In this work we predominantly employ models from the multi-armed bandit (MAB) literature. 
MAB models are applicable to sequential learning problems where one or several players attempt to maximize rewards observed from action choices. 
The MAB model does not provide the player(s) with prior information regarding the reward for each action choice.
When multiple players are included, the relevant models are called multi-player multi-armed bandits (MMAB) \cite{MMAB_survey}. 
MMAB models are a recent development, motivated primarily by cognitive radio networks \cite{5412697} \cite{5535151}. 
Cognitive radio networks are well studied in the literature, but are a very different problem than cognitive radar: while cognitive radio considers channel capacity and optimization for multiple users, cognitive radar attempts to maximize target tracking and detection.  
Further, the parts of a cognitive radio network have individual goals (i.e., desired data rate), while the parts of a cognitive radar network collaborate on joint goals.

MMABs consider multiple independent players acting on a single action set. 
If multiple players select the same action at the same time, they collide and receive a discounted reward. 
Without cooperation, this can turn into a competition between players for the highest-reward actions, causing collisions and generally reducing performance. 
When the players cooperate, they can instead optimize for network-optimal solutions, rather than single-node optimal solutions. 
Further, the presence of coordination or communication can reduce instances of collisions and improve reward payouts over time.

As with centralized and decentralized cognition models, there exist centralized and decentralized MMAB models. 
Decentralized models must exploit collisions to exchange information \cite{5462144} \cite{avner2019multi} \cite{mehrabian20a}, while centralized models have the use of a side channel for communication \cite{komiyama2015optimal} \cite{pmlr-v28-chen13a}.

Models also exist for adversarial environments \cite{howard2022_adversarialconf}, \cite{MultiAdversarial}, where interferer behavior can be chosen in advance by an adversary which knows the cognitive strategy being used by the CRN. 
This models the scenario where an interferer attempts to force the CRN into a poor performing configuration. 
Our current work considers the case where interferers are oblivious to the CRN and do not modify their behavior in response to CRN actions since they are considered primary users.

In general, decentralized CRNs have not been well addressed in the literature. 
Specifically, there has been no study of the relationship between feedback and CRN performance. 
Further, while CRN time and frequency resource allocation has been investigated, there is a lack of study into adaptive models.

\subsection{Problem Summary}
As covered above, the problem where \texttt{SINR} is constant in space and time and the CRN is completely distributed with no feedback has been studied in \cite{howard2022_MMABjournal}. 
We instead consider the case where the spectrum is interference limited \oldbluetext{blue}{in every channel}, and the \texttt{SINR} varies by node and over time due to target motion \oldbluetext{blue}{and range from each} radar node location. 

We study reward models which are applicable to this situation. 
We discuss a model where rewards are based solely on average \texttt{SINR} (as determined by the CC), then provide an approximation which reduces the required feedback. 
Since the rewards are dependent on both the interference power in the environment as well as the relative target range at each node, the nodes can estimate future rewards by separating these two effects. 
The goal of the CRN is to predict the \texttt{SINR} for each channel at each node in the following \revision{Coherent Pulse Interval (CPI)}, taking into account observed interference and estimated target behavior. 
Then, each node in the network attempts to select a channel which maximizes the total reward for the network. 
We will study the amount of feedback that a CC can provide in order to accelerate this learning process. 
\oldbluetext{blue}{In particular, we will demonstrate a trade-off between performance and feedback cost. }

This is a coupled estimation problem; the nodes must simultaneously estimate the channel and the target parameters while avoiding other radar nodes in order to learn the environment.

Collisions occur when more than one radar node transmit in the same channel at the same time since they cause unacceptably high levels of interference at the impacted nodes. 
Importantly, the feedback in our network \oldbluetext{blue}{and algorithm} allows collisions to be largely avoided.

\section{Network Structure}
The general structure of our network is as follows. 
The radar network consists of a set $\mathcal{M}$ of $M$ radar nodes. 
\oldbluetext{blue}{These nodes are distributed uniformly at random throughout an area 10km by 10km. }
While realistic scenarios should include a third spatial dimension, we consider two dimensional space to reduce the simulation complexity. 
Since the algorithms we will discuss only require the position and velocity of each target, we can make this assumption without loss of generality. 
The position of each node $R_m \in \mathcal{M}$, $(m\leq M)$, is denoted as $\mathbf{p}_m = [x_m, y_m]$. 
Since the the nodes can exchange information through the CC, $\mathbf{p} = [\mathbf{p}_1, \mathbf{p}_2, \dots, \mathbf{p}_M]$ is known to all nodes.

The environment is assumed to contain \oldbluetext{blue}{one target, and (as we show later), the channels are sufficient to reliably detect the target. } 
The complete target tracking and detection structure is discussed in a later section. 
The position of the target is denoted as $\mathbf{y}_w = [x_w, y_w]$. 
The set $\mathcal{N}$ contains the $N$ orthogonal channels $C_n$ of equal bandwidth. 
Each radar node is able to transmit one Linear Frequency Modulated (LFM) chirp waveform in a channel $C_n$ which it must select.

The CRN divides time into CPIs and further into Pulse Repetition Intervals (PRIs). 
Each CPI consists of $512$ PRIs, where a single PRI lasts for $0.1$ \texttt{ms} and a single waveform lasts $1$ \texttt{$\mu$s}. 
During each CPI $k$, each radar node $R_m \in \mathcal{M}$ executes the following (roughly synchronized to the CC clock): 
\begin{enumerate}
    \item Select a channel $C_n \in \mathcal{N}$ using a learning algorithm, and transmit a train of \oldbluetext{blue}{512} LFM pulses. 
    \item Receive the waveforms and process the returns to determine estimates\footnote{Recall that estimates are denoted with a hat $\hat{p}$, while true parameters are denoted without. } of: 
    \begin{enumerate}
        \item Target range $\hat{r}_m(k)$. 
        \item Target radial velocity $\hat{\dot{r}}_m(k)$. 
        \item Target angle of arrival $\hat{\theta}_m(k)$. 
    \end{enumerate}
    \item Transmit the target parameter estimates to the CC. 
    \item Receive a state estimate for the target from the CC: 
    \begin{enumerate}
        \item Target position $\hat{\mathbf{y}}_w(k)$. 
        \item Target velocity $\hat{\dot{\mathbf{y}}}_w(k)$. 
    \end{enumerate}
    \item Update target tracking filter\footnote{Filters are maintained at each node and at the CC. }. 
    \item Update the learning algorithm \oldbluetext{blue}{to allow better choice of channel}, requesting and incorporating CC feedback as necessary. 
\end{enumerate}
Concurrently, the CC performs the following functions: 
\begin{enumerate}
    \item Receive target parameter estimates \oldbluetext{blue}{for all targets from each radar node. }
    \item Fuse these measurements to determine a target state estimate \oldbluetext{blue}{for all targets. }
    \item Transmit the target state estimate to the radar nodes. 
    \item Provide channel selection feedback as required. 
\end{enumerate}
We refer to target \emph{parameter} estimates, which are the range, radial velocity, and angle of arrival for each target \emph{from the perspective of each node}. 
Target \emph{state} estimates are the fused estimates provided by the CC, which include position estimates as well as velocity estimates. 
Figure \ref{fig:system} shows \oldbluetext{blue}{a diagram of} this network structure. 

\begin{figure*}
    \centering
    \includegraphics[scale=0.75]{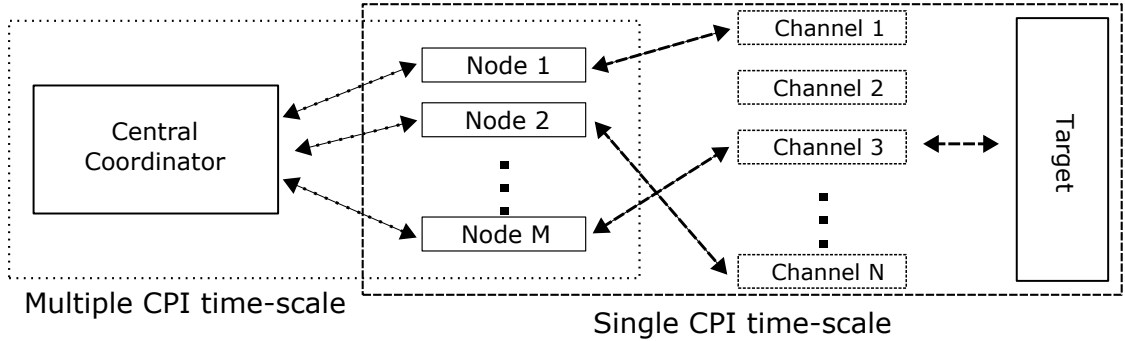}
    \caption{System diagram. Once per CPI, each radar node selects a single channel. Since there are more channels than radar nodes $(N>M)$, some may be unused. However, every radar node will be paired. On a slower timescale (i.e., over multiple CPI's), the radar nodes communicate with and receive feedback from the central coordinator. }
    \label{fig:system}
\end{figure*}

Importantly, we assume that the CRN consists of low cost, low complexity radar nodes. 
This has several implications: 
\begin{itemize}
    \item In order to conserve power, the radar nodes conduct signal processing only once per CPI. 
    \item The transmit arrays of each radar node have sufficient gain to illuminate the target and are electronically steerable. 
    \item Cognition is shared between the CC and nodes to mitigate any duplication of effort. 
\end{itemize}

\subsection{Target and Channel Modeling}
We assume that the environment contains multiple sources in each channel, distributed through space, and sufficient clutter such that the interference has no strong directional components.
This results in interference power with possibly strong variation by channel, but relatively little variation in space. 
For parts of this work, we assume that these spatial variations are sufficient to provide different interference power values at each radar node, but not so much as to cause the \emph{rank} of these values to change. 

\begin{assumption}[Reward Ordering]
\label{assumption:ordering}
    If one radar node observes a greater reward in channel $C_{n_1}$ than in channel $C_{n_2}$, all other nodes will observe the same. The reward magnitudes may however differ. 
\end{assumption}

Note that in some sense, this assumption represents a worst case scenario - while all nodes will observe the same ``best'' channel, only one of them will be able to select it. 
Therefore, in the absence of coordination or feedback, the network would collide frequently.

Later, we will discuss the impact of this assumption, and how it can be relaxed. 
Specifically, we present results with and without this assumption. 

\subsubsection{Signal Model}
In each CPI, each radar node selects a channel $C_n$ with an associated start frequency $f_n$ and transmits a train of $2^{10}$ Linear Frequency Modulated pulse\oldbluetext{blue}{s}. 
Eq. (\ref{eq:pulse}) represents a single pulse. 
\begin{equation}
\label{eq:pulse}
    s[t] = \sin\left[\phi_0 + 2\pi\left(\frac{r}{2}t^2 + f_n t\right)\right], \;\; \oldbluetext{blue}{t \in [t_a, t_b]}
\end{equation}
Here, $t$ is the so-called \emph{fast time} and indexes samples of the pulse, $\phi_0$ is an initial phase, and $r$ is a constant chirp rate.  

The target is modeled as an isotropic scatterer and thus has \oldbluetext{blue}{constant} Radar Cross Section (RCS) as a function of angle-of-arrival $\theta$. 
In addition, we assume that the target \oldbluetext{blue}{response} is not frequency-selective. 
This means that the target will ``look'' the same \oldbluetext{blue}{at all frequencies.} 

We can write the received signal for radar node $R_m$ as Eq. (\ref{eq:rx_model}) where $\tau$ is the propagation delay, $i_n(t)$ is the interference waveform in channel $C_n$ and $n(t)$ is noise.

\begin{equation}
    \label{eq:rx_model}
    y_m[t] =  s\left[(1-\frac{2\dot{r}_m}{c})t - \tau\right] + i_n(t) + n(t)
\end{equation}

Denote the power of the transmitted signal \oldbluetext{blue}{at all nodes} as $P_s = \frac1N \sum_n |s[n]|^2$ and the power of the received signal as Eq. (\ref{eq:rx_total_power}), where $P_{y,m}$ is the power received from the target at node $R_m$, $P_{i,n}$ is the interference power in channel $C_n$, and $\sigma^2$ is the noise power. 

\begin{equation}
    \label{eq:rx_total_power}
    P = \frac1N \sum_n |y_m[n]|^2 = P_{y,m} + P_{i,n} + \sigma^2
\end{equation}
According to the radar equation, the power \oldbluetext{blue}{$P_{y,m}$} should follow Eq. (\ref{eq:rx_power}), where $r_m$ is the target range from the $m^{th}$ node. 

\begin{equation}
    \label{eq:rx_power}   
    P_{y,m} = \frac{P_x G^2\lambda^2\sigma}{(4\pi)^3 r^4_m}
\end{equation}

Since the target RCS is constant over frequencies $f_n$ in the bandwidth we consider and angle $\theta$, it will be constant over radar node measurements. 
Each radar node can form an \emph{estimate} of the future power received from the target as Eq. (\ref{eq:power_estimate}).

\begin{equation}
    \label{eq:power_estimate}
    \hat{P}_m[k_0] = \frac{P_x G^2 \lambda^2}{(4\pi)^3 \hat{r}_m[t+k_0]^4}
\end{equation}

This power estimate depends on an estimate $\hat{r}_m[t+k_0]$ of the range some number $k_0$ of time steps in the future. 
The quality of this estimate will be dictated by the radar observation quality in all time steps until $t$, \oldbluetext{blue}{and is essentially dependent on tracking performance. }

Target estimation quality is directly influenced by channel SINR. 
Denote the SINR experienced by radar node $R_m$ in channel $C_n$ as Eq. (\ref{eq:SINR}). 

\begin{equation}
    \label{eq:SINR}
    \gamma_{m,n} = \frac{P_{y,m}}{P_{i,n} + \sigma^2}
\end{equation}

Since radar measurement quality is influenced by SINR, we'd like to develop a metric which uses this information. So, let the metric be given as Eq. (\ref{eq:channel_metric}) \oldbluetext{blue}{where $\gamma_{m,n}^{(dB)} = 10\log_{10}(\gamma_{m,n})$}.

\begin{equation}
\label{eq:channel_metric}
    \Gamma_{m,n}[t + k_0] = \gamma_{m,n}^{(dB)} - \hat{P}^{(dB)}_{m}[k_0]
\end{equation}

This metric is useful because it allows each radar node $m$ to arrive at a similar estimate of the quality of channel $n$. 
This is necessary due to the distributed nature of the problem; we'd like for the independent radar nodes to be able to avoid colliding with each other (i.e., selecting the same action simultaneously) without communication. 
Prior work has shown that collisions greatly reduce the performance of a radar network \cite{howard2021_multiplayerconf}. 

Note that $\hat{P}^{(dB)}_{m}[k_0] = 10\log_{10}(\hat{P}_m[k_0])$. 
Due to the assumption on interference power ordering, we can now see that if $R_{m_1}$ experiences $P_{m_1,n_1} > P_{m_1,n_2}$ for two channels $C_{n_1}$ and $C_{n_2}$, then $R_{m_2}$ will observe the same power ordering ($P_{m_2,n_1} > P_{m_2,n_2}$) for any choice of radars $R_{m_1}, R_{m_2}$ or channels $C_{n_1}, C_{n_2}$. 

\subsection{Tracking Formulation}
While the spectrum is interference-limited, we assume that the best case channels have SINR high enough for consistent target detection. 
\oldbluetext{blue}{Figure \ref{fig:PdPfa} shows the probability of detection versus probability of false alarm for the best, typical, and worst case network average \texttt{SINR} based on the assumed parameters in Table \ref{tab:params}. }

%
\begin{figure}
    \centering
    \includegraphics[scale=0.6]{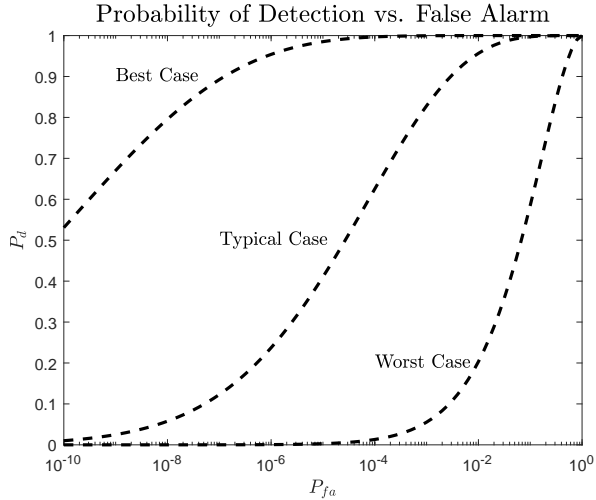}
    \caption{Probability of detection versus probability of false alarm for the worst, average, and best case node-channel matchings. These values are based on the assumed parameters stated in Table \ref{tab:params}. }
    \label{fig:PdPfa}
\end{figure}

When each radar node observes the target, it estimates the target position $\hat{\mathbf{y}}_m(k)$ and velocity $\hat{\dot{\mathbf{y}}}_m(k)$. 
These are used to update a Kalman filter model of the target's motion. Assuming a two-dimensional motion model, the predicted state is given as Eq. (\ref{eq:predict}) where $F_k$ is the transition model. 
Note that the state $x_k = [\mathbf{y}_m(k)^T, \dot{\mathbf{y}}_m(k)^T]^T$ is composed of target position and velocity. 

\begin{equation}
    \label{eq:predict}
    \hat{x}_{k|k-1} = F_k x_{x-1|k-1} + B_k u_k
\end{equation}
\revision{The terms $B_k$ and $u_k$ correspond to the control input model and corresponding control vector \cite{richards2014principles}. }
The state is then updated as Eq. (\ref{eq:update}) where $K_k$ is the \oldbluetext{blue}{estimated optimal} Kalman gain and $\tilde{y}_k$ is the innovation. 
\begin{equation}
    \label{eq:update}
    x_{k|k} = \hat{x}_{k|k-1} + K_k \tilde{y}_k
\end{equation}
This formulation follows the common notation of \cite{hamilton2020time}. 

\section{Learning Structure}

As is common in the MAB and sequential learning literature, we will define our problem based on a series of \emph{actions} taken by the players (i.e., radar nodes) and \emph{rewards} provided by the environment. 
We are specifically considering a \emph{sensing then collision} \cite{besson2018multi} model. 
This means that the players observe first the prospective reward (discounted, in case of collision) for a given action, followed by information on collisions. 
Collisions are instances of two radar nodes choosing identical actions at the same time. 
However, due to the learning framework and presence of feedback, collisions are unlikely. 
This is a realistic model since previous work has demonstrated a method to detect collisions \cite{howard2022_MMABjournal}. 
In addition, we make the assumption rewards observed by node $R_{m_1}$ are not available to any other node $R_{m_2}$. 
This follows the assumed network structure; nodes can exchange information with the CC but not directly with each other.

\subsection{Matchings and Utility}
Before we can define the learning framework, we need to better \oldbluetext{blue}{describe} the objective. 
Let a \emph{weight matrix} $W$ be \oldbluetext{blue}{a matrix with $M$ rows corresponding to the $M$ radars and $N$ columns corresponding to the $N$ channels. }
Each index consists of the reward\footnote{The specific reward function is defined in the following section. } observed by radar $R_m$ for selecting channel $C_n$ during the $k^{th}$ CPI. 
Valid actions which can be selected by algorithms under our framework must be in the set of all \emph{matchings}. 
\begin{definition}[Matching]
    A \emph{matching} $\pi: \mathcal{M} \to \mathcal{N}$ is any assignment from the set of radar nodes $\mathcal{M}$ to the set of channels $\mathcal{N}$ with the following properties: 
    \begin{enumerate}
        \item Matchings are injective\footnote{While they are injective, matchings are not necessarily bijective since $M$ does not necessarily equal $N$. A function $f:X\to Y$ is injective if for all $a,b$ in $X$, $f(a)=f(b)\implies a=b$.  }. \oldbluetext{blue}{This means that every radar node will select a single channel per CPI, but not every channel will necessarily be used. }
        \item Matchings map every element of $\mathcal{M}$ to a unique element of $\mathcal{N}$. 
    \end{enumerate}
    Denote the set of all matchings as $\Pi$. 
\end{definition}
Matchings are a special case of \emph{mappings}, which remove the injectiveness property. 
Obviously, since the weights vary by radar node and by channel, some matchings will be better than others. 
We measure the quality of a matching via its \emph{utility}. 
\begin{definition}[Utility]
    The \emph{utility} of a matching $\pi$ is the sum of the rewards observed under that matching. 
    \begin{equation}
        U(\pi) = \sum_{\mathcal{M}} W_{m, \pi_m}
    \end{equation}
\end{definition}
The utility of a matching represents the quality of each channel selected for radar observation. 
Note that $W_{m, \pi_m}$ refers to reward observed by $R_m$ due to selecting a channel $C_{\pi_m}$, where $\pi_m$ is the index of matching $\pi$ corresponding to node $R_m$.  
Further, there will be at least one $\pi \in \Pi$ with greatest utility. 
utility describes the quality of measurement obtained by a particular node for selecting a given channel. 
\begin{definition}[Optimal Matching]
    If a matching $\pi\in \Pi$ has maximum utility $U^*$, it is called optimal and denoted $\pi^*$. 
    In other words, $U(\pi) = U^*\implies \pi=\pi^*$. 
    \begin{equation}
        U^* = \max_{\pi\in\Pi} U(\pi)
    \end{equation}
\end{definition}
\begin{remark}
Note that while there may be many $\pi\in\Pi$ with $U(\pi)=U^*$, we slightly abuse notation and simply refer to any optimal matching as $\pi^*$. 
In practice, there is very rarely more than one optimal matching for a given weight matrix. 
\end{remark}

\subsection{Learning Objective} 
The goal of any learning algorithm in this system is then to minimize the amount of non-optimal matchings encountered during a game. 
Of course, since there is no \emph{a priori} knowledge of the environment, it is impossible to avoid selecting non-optimal matchings or even to know the value of $U^*$. 
\oldbluetext{blue}{This is particularly important in radar problems since sub-optimal mappings can result in missed targets. }

It is useful to view a learning algorithm as a function $\mathfrak{f}(\mathbb{E}) = \pi^{(K)}$ which produces a sequence of matchings $\pi^k$, $k = 1:K$ in an environment $\mathbb{E}$ where $K$ is some finite horizon. 
Here, $\pi^k$ is a single matching while $\pi^{(K)}$ denotes a sequence. 
Then, the sequence $\pi^{(K)}$ contains all of the matchings produced by the learning algorithm until CPI $K$. 
Note that $\pi^{(K)}$ is implicitly conditioned on a specific instance of an environment; if the sequence of rewards changes, then the sequence of matchings would change. 
Also note that $\pi^{(K)}_m$ is the slice of actions chosen by radar node $R_m$ until time horizon \revision{$K$}.

We can measure the difference between learning algorithms by comparing the cumulative utility of a matching sequence until CPI \revision{$k<K$}. 
\begin{equation}
    \label{eq:cum_utility}
    U^k(\pi) = \sum_{\kappa=1}^k U(\pi^{\kappa})
\end{equation}
In order to compare all learning algorithms to a universal baseline, we can refer to the utility of the sequence of \emph{optimal} matchings $\pi^*(k)$ for a given environment. 
This quantity is called the \emph{cumulative regret} of $\pi^k$. 
\begin{definition}[Cumulative Regret]
    The \emph{cumulative regret} of a learning algorithm $\mathfrak{f}$ which produces a sequence of matchings $\pi^k$ until time $k$ is the difference in cumulative utility between $U^k(\pi)$ and $U^k(\pi^*)$. 
    \begin{equation}
        \label{eq:cum_regret}
        \rho^t_\mathfrak{f} = U^k(\pi^*) - U^k(\pi)
    \end{equation}
\end{definition}
Note that cumulative regret is monotonically increasing in $k$, since $U(\pi^k) \leq U(\pi^{*,k})$ by definition.

Now, the \emph{objective} of a learning algorithm $\mathfrak{f}$ is to obtain the lowest $\rho^K_\mathfrak{f}$ for some finite time horizon $K$. 

\subsection{Rewards}
The learning problem is not fully defined without specifying the reward function. 
Typically, sequential learning rewards are drawn from some distribution, dependent on the action selected by the learner. 
We will define two different reward functions that capture key aspects of the radar scenario.

The key differentiation between the two reward functions we will describe is an explicit separation between the two underlying estimation processes. 
The \emph{interference estimation} process is the part of the cognitive radar scenario where each node attempts to learn some metric of the interference in each channel. 
The \emph{target estimation} process, however, is the overall goal of the cognitive radar network. 
In the absence of interference estimation, the network may select poor channels over time and therefore sacrifice tracking performance. 
However, if the network attempts to optimize too quickly for radar tracking performance, again it may suffer from selecting sub-optimal actions. 
This is the classic trade-off in sequential learning between \emph{exploration} and \emph{exploitation}. 
The cognitive nodes must efficiently balance the exploration and exploitation in order to avoid sub-optimal long-term performance.

We will first show a reward function that attempts to separate these underlying processes, and then discuss a simpler model.

\subsubsection{\texttt{SINR} Rewards}
The first rewards we consider are based solely on the \texttt{SINR} observed by radar node $R_m$ in channel $C_n$. 
Due to variability in the environment (i.e., target motion or changes in interference), the \texttt{SINR} may change from CPI to CPI. 
We define the \emph{true} \texttt{SINR} observed by node $R_m$ in channel $C_n$ as $\gamma_{m,n}(k)$ and note that this value will vary by node and channel due to relative spatial differences in target position and differences in interference. 
Let the full matrix of these values be denoted as $\gamma(k)$ in a CPI $k$. 
Let each node draw an \emph{estimate} of this \texttt{SINR} as Eq. (\ref{eq:sinr_est}). 
\begin{equation}
    \label{eq:sinr_est}
    \hat{\gamma}_{m,n}(k) \sim \mathcal{N}(\gamma_{m,n}(k), \sigma^2_\gamma)
\end{equation}
%
Now we can form each element of the weight matrix under this reward function in CPI $k$ as $W^{\hat{\gamma}(k)}_{m,n} = \hat{\gamma}_{m,n}(k)$. 
Also, we can write the utility of a matching $\pi^k$ under \texttt{SINR} rewards as $U^{\hat{\gamma}}(\pi^k)$. 
\begin{equation}
    \label{eq:sinr_utility}
    U^{\hat{\gamma}}(\pi^k) = \sum_{m\in \mathcal{M}} W^{\hat{\gamma}(k)}_{m, \pi_m}
\end{equation}

\subsubsection{Target Based Rewards}
In practice, \texttt{SINR}-based rewards as described above would require each node to share its observed rewards with the CC, and then to rely on the CC to provide actions. 
This is because one node would have no other way to know the rewards being experienced by another node. 
This reduces the redundancy of the system, since the central coordinator is the only agent making decisions. 
If the radar nodes were instead able to estimate the rewards \oldbluetext{blue}{observed by each} other node, then they would be able to make decisions in the absence of the coordinator.

As shown previously, the channel metric Eq. (\ref{eq:channel_metric}) attempts to decouple the interference behavior from the target motion. 
Following the assumption that interference behavior in each channel is identical \oldbluetext{blue}{as observed by each node}, we can recombine the channel metric with an estimate of the target range at each node to estimate the \texttt{SINR} observed at each location in the network. 
In other words, the channel metric combined with an estimate of the target position can produce a reward estimate \oldbluetext{blue}{while requiring} less information than the matrix $W^{\hat{\gamma}(k)}$.

Now, we can write the elements of this new reward function as $W^{\Gamma}(m,\pi_m) = \frac{1}{\hat{r}_m^4}\Gamma_{m,n}$ or more generally as Eq. (\ref{eq:target_rewards}) where $\overline{\mathbf{r}} = [\hat{r}_1, \hat{r}_2, \dots, \hat{r}_M]$ is a vector of the estimated distance from each node to the target and $\overline{\Gamma}_m =[\Gamma_{m,1}, \Gamma_{m,2}, \dots, \Gamma_{m,N}]$ is a vector of channel metrics calculated by node $R_m$. 
Note that $W^{\Gamma}$ is a $M\times N$ matrix\oldbluetext{blue}{, and we later denote indices with subscripts. }
The channel metric is divided by the estimated distance to each node in order to favor those nodes with better views of the target. 
\begin{equation}
    \label{eq:target_rewards}
    W^{\Gamma} = \left(\frac{1}{\overline{\mathbf{r}}^4}\right)^T \odot \overline{\Gamma}
\end{equation}

Then, the estimated utility under this reward function is expressed as Eq. (\ref{eq:target_utility}). 
\begin{equation}
    \label{eq:target_utility}
    U^{\Gamma}(\pi^k) = \sum_{m\in\mathcal{M}}W^{\Gamma}_{m, \pi_m^k}
\end{equation}

\begin{lemma}[Reward Equivalency]
\label{lem:equivilence}
    The optimal matching under SINR rewards is equal to the optimal matching under target-based rewards when \oldbluetext{blue}{Assumption \ref{assumption:ordering}} holds. 
    \begin{equation}
    \max_{\pi\in\gamma(k)}U(\pi) = \max_{\pi\in W^{\Gamma}(k)}U(\pi)
    \end{equation}
\end{lemma}
\begin{proof}
    See Appendix. 
\end{proof}

Of course, if the optimal matching provided by each reward function is the same as in Lemma \ref{lem:equivilence}, why should a CRN prefer one reward function over the other? 
A single node, without coordination or feedback, can not know the rewards observed by another node without feedback. 
This means that nodes may not be able to establish a consensus. 
However, if the node is able to calculate the channel metric (which is node-independent) and target position, it can then estimate the rewards observed by all other nodes. 
As we will show later, this can allow a CRN to develop a near-optimal matching, while \oldbluetext{blue}{receiving} feedback at semi-regular intervals can further improve this performance. 

\subsection{Feedback}
In addition to measuring the performance of an algorithm through regret, we can analyze the amount of information that algorithm exchanges through the CC. 
In particular, we can look at the average number of floating-point values sent from the CC to each node. 
This will allow us to compare the benefits of varying levels of feedback in the network. 
Call $F_k$ the set of values transmitted by the CC in CPI $k$. 

\begin{definition}[Average Feedback]
    The \emph{average feedback} used by a network in a CPI $k$ is the sum of all feedback $|F_j|$ until $k$ divided by $k$ and the number of nodes $M$. 
    \begin{equation}
    \label{eq:feedback}
        F_a(k) = \frac{1}{Mk}\sum_{j=1}^k |F_j|
    \end{equation}
\end{definition}

\section{Candidate Algorithms}
We present several algorithms, starting with a fully centralized variant and moving towards minimal feedback. 
We do this to investigate the trade-off between feedback and performance. 
Ultimately, we show that a minimal amount of feedback is sufficient to provide dramatic benefits over the prior art, and increasing feedback provides diminishing returns.

Several of the algorithms are based on Explore Then Commit \cite{mehrabian20a}. 
This was initially developed for fully decentralized action selection. 
We make a slight modification to allow for a central coordinator to eliminate a lengthy phase where nodes exchange information.

We will describe an oracle for this problem, which is aware of the true \texttt{SINR} in each time step, and perfectly selects the optimal matching. 
In addition we compare against a naive algorithm which selects a new random matching each CPI as well as a previously proposed decentralized algorithm \cite{howard2022_MMABjournal}.

The algorithms we discuss are summarized in Table \ref{tab:algos}. 
\revision{Since several of these algorithms, though modified for our use, are drawn from the literature, we provide no detailed analysis of their computational complexity. 
Instead we provide references to the source literature. 
In addition, the theoretical regret bound for each algorithm is provided in Table \ref{tab:algos}. }

\subsection{Oracle}
An oracle for this problem knows the true \texttt{SINR} and target position, and thus can perfectly estimate the rewards.  
Therefore, in each CPI $k$, the each node $R_m$ in the network will select $\pi^*_m(k)$. 
This ensures that the cumulative regret for the oracle is always 0. 

\subsection{Centralized Explore-Then-Commit \oldbluetext{blue}{(C-ETC)}}
C-ETC \cite{mehrabian20a} implements the Upper Confidence Bound (UCB) \cite{UCB_fischer} as a threshold to refine a sequence of sets of matchings $\Pi_0, \Pi_1$, etc., with $\Pi = \Pi_0 \supseteq \Pi_1 \supseteq \Pi_2 \supseteq$ etc. 
Each set of matchings $\Pi_j$ can be said to contain $p_j$ matchings and is viewed as an ordered list. 
Each radar node $R_m$ selects channel $\pi_{l,m}\in \pi_l$, and $\pi_l\in\Pi_j$, with $l, j$ specified by C-ETC. 
Since each $\Pi_j$ is a refinement of $\Pi_{j-1}$, $\Pi_0$ can be known to all of the nodes since it is the initial condition. 
Then for each subsequent $\Pi_j, j>0$, the CC can indicate which matchings remain and which are discarded. 
Thus, all nodes know the channels used by other nodes.

This allows each node to observe rewards in the environment sequentially and possibly multiple times\oldbluetext{blue}{, to allow the observed rewards to average towards the true rewards}. 
Over time, the algorithm identifies which matchings $\pi_l\in\Pi_j$ have $U(\pi_l)$ below a threshold, and removes them from $\Pi_{j+1}$. 
\oldbluetext{blue}{This is a fully-centralized algorithm, since decisions are only made at the CC. }
A sketch of our implementation of C-ETC is shown in Algorithm \ref{algo:C-ETC}. 

\vspace{0.1in}
\begin{algorithm}
\SetAlgoLined
\% CPI k\;
Transmit radar waveform in channel $C(k) = \pi_{l,m}\in \pi_l \in \Pi_j$\;
Estimate target parameters $\hat{r}_m(k), \dot{\hat{r}}_m(k), \hat{\theta}_m(k)$\;
 \uIf{$|\Pi_k| = 1$}{
  Set $\Pi_{j+1} = \Pi_j$\;
  }
 \Else{
  \uIf{$l = |\Pi_j|$}{
   Transmit target estimates and channel estimates to CC\;
   $l = 0$\; 
   Receive $\Pi_{j+1}$ from CC\;
   $j = j + 1$\;
  }
  $l = l + 1$\;
 }
$k = k + 1$\;
 \caption{Sketch of Explore-Then-Commit for node $R_m$}
 \label{algo:C-ETC}
\end{algorithm}
\vspace{0.1in}
In this way, the action selection is delegated to the CC, while the nodes conduct the radar processing. 
The CC also combines radar observations from the network to determine a final target state estimate. 

\subsection{Centralized Explore-Then-Predict \oldbluetext{blue}{(C-ETP)}}
C-ETP modifies C-ETC in one major way\oldbluetext{blue}{, while remaining a centralized algorithm}. 
Rather than committing to a single matching after exploration, C-ETP continues to evaluate the weight matrix until the end of the game. 
From a given matrix, C-ETP calculates the maximum matching at each node and selects it. 
The weight matrix is formed in each CPI $k$ at the CC as $W^{\hat{\gamma}(k)}$ by using the SINR measurements made at each node. 
This results in an estimator that is able to modify action selections over time in order to adapt to the changing environment. 
\oldbluetext{blue}{Since each node is acting on the same information, they can know which actions the other nodes will select. }
This has the relatively large downside of requiring feedback from the CC in every time step, greatly \emph{increasing} the feedback rate after convergence. 
\oldbluetext{blue}{C-ETC does not suffer from this problem since it does not evaluate the reward matrix after convergence. }

The C-ETC-like exploration period is still required for C-ETP due to the noisy observations; without a well-defined exploration structure, the algorithm would fail to converge at all. 
However, as the nodes continue to utilize the best performing channels, the weights may cause two nodes to switch channels due to target position. 

A sketch of the node-side algorithm for C-ETP is shown in Algorithm \ref{algo:C-ETP}. 
Note that if more than one matching in $W^{\hat{\gamma}(k)}$ is optimal, the algorithm selects one at random. 

\revision{The regret bound for ETC based algorithms is $\mathcal{O}\left(\ln T\right)$. 
Since regret represents the distance between optimal actions and the ones selected by an algorithm, we can interpret the regret \emph{bound} as an indicator of duration, i.e., how long the algorithm must run before it ``converges'' to the optimal solution. 
Further, it is shown in \cite{bistritz2021game} that a regret bound of $\mathcal{O}\left(\ln T\right)$ is optimal for this type of algorithm. 
So, the regret bound for ETC based algorithms is optimal with respect to convergence times.}

\begin{table*}[t]
    \centering
    \caption{Candidate Algorithms}
    \begin{tabular}{||c |c |c |c |c |c ||}
    \hline 
    Algorithm & Acronym & Description & Rewards & Feedback & \revision{Regret Bound}\\
    \hline \hline
    Oracle & N/A & Selects optimal matchings. & N/A & None. & \revision{N/A}\\
    \hline 
    Explore-Then-Commit & C-ETC & Full feedback. & SINR & Matching sequence. & \revision{$\mathcal{O}(\ln T)$}\\
    \hline 
    Centralized Explore-Then-Predict & C-ETP & Full feedback. & SINR & Entire reward matrix. & \revision{$\mathcal{O}(\ln T)$}\\
    \hline 
    Hybrid Explore-Then-Predict & H-ETP & Hybrid cognition. & Target-based & Current target state. & \revision{$\mathcal{O}(\ln T)$}\\
    \hline
    Explore-Then-Predict & ETP & Decentralized after setup. & Target-based & Initial matchings. & \revision{$\mathcal{O}(\ln T)$}\\
    \hline
    Musical Chairs \cite{howard2022_MMABjournal} & MC & Completely decentralized. & SINR & None. & \revision{$\mathcal{O}(\ln xT)$}\\
    \hline
    Random Matchings & N/A & Selects random matchings. & N/A & None. & \revision{N/A}\\
    \hline
    \end{tabular}
    
    \label{tab:algos}
\end{table*}
\vspace{0.1in}
\begin{algorithm}
\SetAlgoLined
\% CPI k\;
Transmit $C(k) = \pi_{l,m}\in \pi_l \in \Pi_j$\;
Estimate $\hat{r}_m(k), \dot{\hat{r}}_m(k), \hat{\theta}_m(k)$\;
 \uIf{$|\Pi_k| = 1$}{
  Receive $W_{\hat{S}(k)}$ from CC\;
  $\Pi_{j+1} = \max_{\pi \in \Pi(W_{\hat{S}(k)})}U(\pi)$\;
  }
 \Else{
  \uIf{$l = |\Pi_j|$}{
   $l = 0$\; 
   Receive $\Pi_{j+1}$ from CC\;
   $j = j + 1$\;
  }
  $l = l + 1$\;
 }
$k = k + 1$\;
 \caption{Sketch of Centralized Explore-Then-Predict for node $R_m$}
 \label{algo:C-ETP}
\end{algorithm}
\vspace{0.1in}

\subsection{Hybrid Explore-Then-Predict \oldbluetext{blue}{(H-ETP)}}
H-ETP alters C-ETP to use the target-based rewards $W^{\Gamma}$ Eq. (\ref{eq:target_rewards}) to make action decisions after convergence. 
This only requires knowledge of the target range from each node. 
So, rather than transmit all of the estimated rewards from each node, H-ETP only requires the current estimated target state to be transmitted from the CC to the nodes. 
This greatly reduces the feedback rate\oldbluetext{blue}{, thus, this algorithm uses hybrid cognition rather than centralized. }

In addition, due to the structure of the target-based rewards and using Lemma \ref{lem:equivilence}, the performance of H-ETP will at least match that of C-ETC with a greatly reduced feedback rate.

\vspace{0.1in}
\begin{algorithm}
\SetAlgoLined
\% CPI k\;
Transmit $C(k) = \pi_{l,m}\in \pi_l \in \Pi_j$\;
Estimate $\hat{r}_m(k), \dot{\hat{r}}_m(k), \hat{\theta}_m(k)$\;
 \uIf{$|\Pi_k| = 1$}{
  Receive $\hat{\mathbf{y}}_w(k)$ from CC\;
  $\hat{\mathbf{r}} = |\mathbf{p} - \hat{\mathbf{y}}_w(k)|$\;
  $W_{P(k)} = \frac{1}{\hat{\mathbf{r}}} * \overline{\mathbf{P}}$\;
  $\Pi_{j+1} = \max_{\pi \in \Pi(W_{P(k)})}U(\pi)$\;
  }
 \Else{
  \uIf{$l = |\Pi_j|$}{
   $l = 0$\; 
   Receive $\Pi_{j+1}$ from CC\;
   $j = j + 1$\;
  }
  $l = l + 1$\;
 }
$k = k + 1$\;
 \caption{Sketch of Hybrid Explore-Then-Predict for node $R_m$}
 \label{algo:HETP}
\end{algorithm}
\vspace{0.1in}

\subsection{Explore-Then-Predict \oldbluetext{blue}{(ETP)}}
ETP extends this line of thinking (reducing feedback) even further; instead of relying on the network's estimated target state, ETP simply uses the target parameters estimated at each node to estimate the range to each node. 
Due to compounding estimation errors, ETP should have reduced tracking accuracy and \oldbluetext{blue}{higher} regret than either H-ETP or C-ETP, but require the smallest feedback rate. 
\oldbluetext{blue}{Specifically, the only information ETP requires is in the initial part of the game. 
The CC must determine initial matchings to explore in order to prevent collisions. }

\vspace{0.1in}
\begin{algorithm}
\SetAlgoLined
\% CPI k\;
Transmit $C(k) = \pi_{l,m}\in \pi_l \in \Pi_j$\;
Estimate $\hat{r}_m(k), \dot{\hat{r}}_m(k), \hat{\theta}_m(k)$\;
 \uIf{$|\Pi_k| = 1$}{
  $\hat{\mathbf{y}}_{w,m}(k)$\;
  $\hat{\mathbf{r}} = |\mathbf{p} - \hat{\mathbf{y}}_{w,m}(k)|$\;
  $W_{P(k)} = \frac{1}{\hat{\mathbf{r}}} * \overline{\mathbf{P}}$\;
  $\Pi_{j+1} = \max_{\pi \in \Pi(W_{P(k)})}U(\pi)$\;
  }
 \Else{
  \uIf{$l = |\Pi_j|$}{
   $l = 0$\; 
   Receive $\Pi_{j+1}$ from CC\;
   $j = j + 1$\;
  }
  $l = l + 1$\;
 }
$k = k + 1$\;
 \caption{Explore-Then-Predict for node $R_m$}
 \label{algo:ETP}
\end{algorithm}
\vspace{0.1in}
\subsection{Random Matchings}
We finally consider a naive algorithm which simply selects from a pre-determined random sequence of matchings. 
Before the game begins, let the CC specify an appropriately sized set $\Pi^R$ of matchings. 
In each CPI $k$, let node $R_m$ select channel $\Pi^R_m(k)$. 
This ensures that two nodes never select the same channel in a CPI. 
However, since the matching sequence is pre-determined, the network is not able to learn anything about the environment and simply experiences an average of the possible performance. 
In addition, since the matching sequence is random and not necessarily optimal in any CPI $k$, we should expect roughly linear regret in time. 

\begin{remark}
    Since $\gamma(k)$ (the full matrix of true \texttt{SINR} values) varies in time depending on the target location, algorithms which continue to evaluate their reward matrix will attain lower regret and superior tracking performance than those which retain a fixed allocation post convergence. 
\end{remark}

\subsection{Musical Chairs (MC)}
Musical Chairs is an algorithm developed in \cite{besson2018multi} and applied to CRNs in \cite{howard2022_MMABjournal} for the completely decentralized case. It relies on a system of implicit collisions through which the network establishes the best-case matching available. 
This causes a large amount of regret prior to convergence, which can take up to $10^3$ time steps (much longer than other algorithms considered here).


\section{Simulations}
\begin{table*}[t]
    \centering
    \caption{Simulation parameters, unless stated otherwise. }
    \begin{tabular}{||c | c||c | c||}
    \hline 
    Parameter & Value & Parameter & Value\\
    \hline \hline
    Number of Radar Nodes $M$     & 5 & Number of Targets & 1\\
    \hline 
    PRIs per CPI   & 500 & Target Initial Position & [0,0] \texttt{m}\\
    \hline 
    Total CPIs  & 700 & Bandwidth & 20 \texttt{MHz}\\
    \hline 
    Typical \texttt{SINR} & 12 \texttt{dB} & Averaged Simulations & 30\\
    \hline
    Frequency & 2.4 \texttt{GHz}  & PRI Duration & $1.024\times10^{-4}$ \texttt{s}\\
    \hline
    Transmit power & 20 \texttt{dBw} & RCS & 100 \texttt{m}$^2$\\
    \hline 
    Antenna gain & 30 \texttt{dB} &  &  \\
    \hline
    \end{tabular}
    
    \label{tab:params}
\end{table*}

We simulate a CRN with five radar nodes and a CC which collaboratively track a single target. 
The radar nodes are distributed randomly through an area 10 kilometers by 10 kilometers. 
The single target is initially located at the origin, and moves with a velocity of 200\texttt{m/s} headed northeast. 
The target has a uniform radar cross section of $100\texttt{m}^2$. 
These values are consistent with a typical commercial aircraft \cite{knott2006radar}.

The radar nodes have access to eight equally spaced channels of 20\texttt{MHz} each from $2.34-2.5\texttt{GHz}$. 
Each transmitter outputs pulses at 20\texttt{dBw}, and the arrays have a main beam gain of 30\texttt{dB}. 
These and other simulation parameters are available in Table \ref{tab:params}. 
Fig. \ref{fig:scenario} shows a single instance of this scenario. 
\begin{figure}
    \centering
    \includegraphics[scale=0.6]{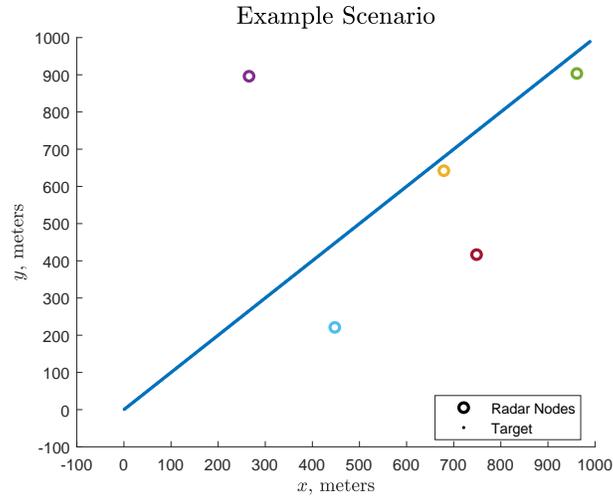}
    \caption{
    The spatial distribution of the radar nodes and the path of the target. 
    Radar positions are drawn from a uniform distribution in each simulation instance. 
    As the target moves through the scene, different radars benefit from selecting different actions. }
    \label{fig:scenario}
\end{figure}

Generally, we should expect to see algorithms with higher feedback exhibit greater performance. 

As discussed above, each algorithm will incur regret due to sub-optimal channel selection. 
In general, we should expect lower regret for algorithms which more closely model the environment. 
In Fig. \ref{fig:regret}, we see that the cumulative average regret of each learning algorithm goes to zero over time. 
However, the regret for the random matching algorithm remains constant in time. 
\oldbluetext{blue}{The learning algorithms take roughly 100 CPIs to converge. 
The decentralized algorithm MC takes much longer than the simulation time to converge (on the order of $10^3$ time steps) and thus exhibits performance roughly equal to random matchings over this short time horizon. }

\begin{figure}
    \centering
    \includegraphics[scale=0.6]{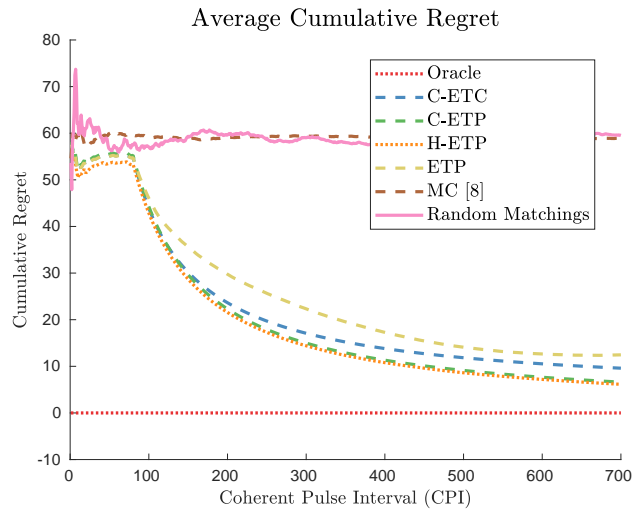}
    \caption{
    After convergence (around CPI 100), the regret of each algorithm goes towards zero. This is because each algorithm is able to identify the best channels to select. 
    However, this identical regret performance does not translate to identical radar tracking performance. }
    \label{fig:regret}
\end{figure}

It is also important to consider the feedback rate for each algorithm. 
As shown in Eq. \ref{eq:feedback}, we can measure the amount of feedback by counting the number of floating-point numbers transferred from the CC to the nodes in each time step. 
Fig. \ref{fig:feedback} shows this. 
Since all of the learning algorithms incorporate an C-ETC-like exploration phase, they all require a similar amount of information in the early part of the game. 
However, once the algorithms converge, this feedback rate begins to change. 

\oldbluetext{blue}{
After convergence (around $100$ CPIs), C-ETC and ETP do not use any feedback, as explained above. 
Before convergence, feedback is only required after the network has explored each list of matchings. 
Therefore, the average feedback rate trends towards zero. 
MC uses no feedback at all and thus is not plotted. 
C-ETP, on the other hand, uses a substantial amount of feedback in each CPI to maintain an updated reward matrix. 
Therefore, the feedback rate increases after convergence. 
Finally, H-ETP uses a small amount of feedback every CPI to maintain knowledge of the target position, using this information to estimate the reward matrix. 
Instead of transferring the entire weight matrix in every CPI, the use of the channel metric allows the H-ETP CC to only transmit the predicted target location in each CPI. 
This reduction in feedback balances cognition between the CC and the nodes. 
}

\oldbluetext{blue}{
ETP uses no feedback after initialization since it relies on node-specific target location estimates to predict the reward matrix. 
Since these internal estimates are less precise, the predicted reward matrix will be less accurate (possibly leading to collisions). 
While not impossible, collisions are very rare events. }

\begin{figure}
    \centering
    \includegraphics[scale=0.6]{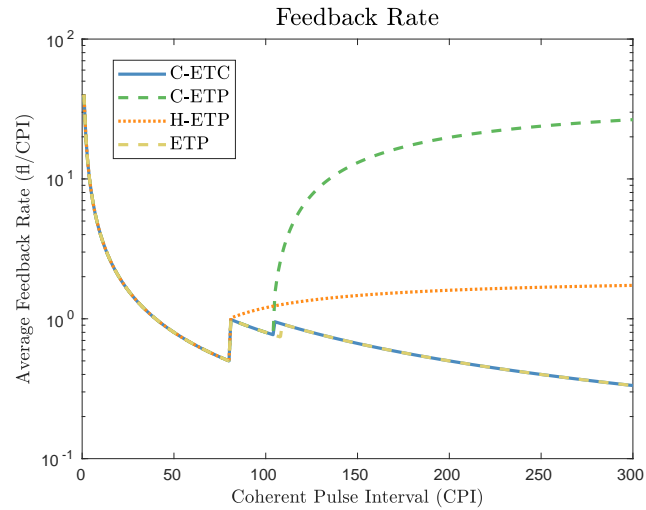}
    \caption{
    The average feedback per node in each of the different networks we examine. 
    ETP is shown to not use any feedback after convergence, while C-ETP uses a great deal of feedback. }
    \label{fig:feedback}
\end{figure}

The radar tracking performance of each algorithm is also important in our application. 
Since the learning algorithms must explore the environment in the early parts of each simulation, we should expect higher tracking error. 
In addition, since we use a Kalman tracking filter, we should see lower error once the filter converges. 
In Fig. \ref{fig:ave_error}, we see that the average error for the learning algorithms is consistently below 10\texttt{m}. 
Occasionally, the environment will shift and the error will spike. 
Each CRN uses the observation from all radar nodes to establish a localization estimate once per CPI, which we compare against the target's true location in the middle of the CPI. 
We can see that the oracle exhibits the best performance for the entire simulation, while Random Matchings has quite variable performance. 
Also note the increased error at the beginning and end of the track. 
This is due to the target being relatively further away from the nodes during these periods, as well as the poor geometry. 
\begin{figure}
    \centering
    \includegraphics[scale=0.6]{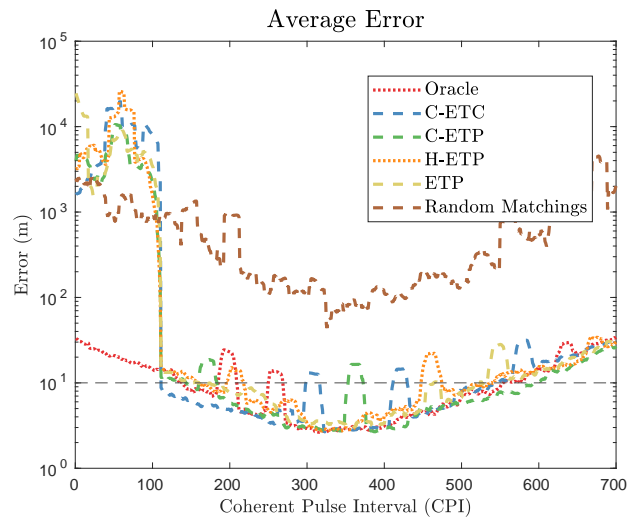}
    \caption{
    Radar localization error, averaged over 30 simulations. 
    The performance of H-ETP is shown to be slightly reduced from C-ETP, but still superior to the other techniques. 
    The increased error in the beginning of the simulation is due to convergence time, both of the machine learning algorithm and of the Kalman tracking filter. }
    \label{fig:ave_error}
\end{figure}

In addition to the average error, we can look at the distribution of the error. 
In Fig. \ref{fig:ecdf_full}, we see an empirical CDF for the error of each algorithm. 
\begin{figure}
    \centering
    \includegraphics[scale=0.6]{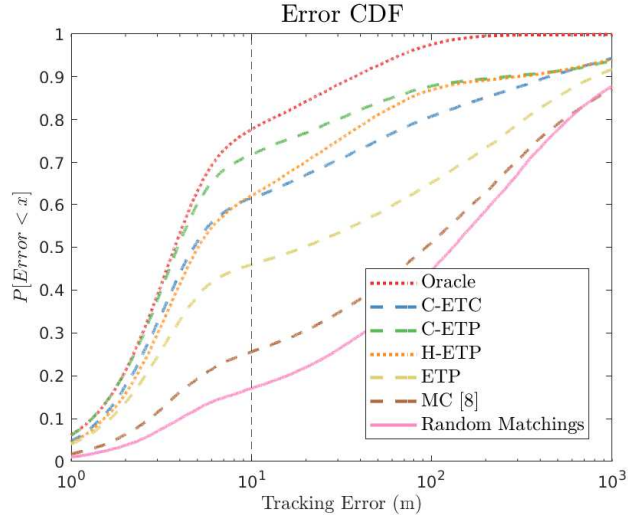}
    \caption{
    The error distribution for each algorithm for the entire simulation. }
    \label{fig:ecdf_full}
\end{figure}
This represents the probability that the error will be less than a given value. 
In general, we would like there to be a high probability that the error is small. 
So, curves which are further to the left are better. 
From this, we can see that the performance of Random Matching is far below the learning algorithms. 
\oldbluetext{blue}{In addition, the performance of C-ETP nearly reaches that of the oracle. 
This is because C-ETP has a constantly-updated reward matrix containing information from the entire network. 
Of course, as we saw above, this has a high feedback cost. }

\oldbluetext{blue}{
ETP has reduced performance due to inaccurate prediction of the reward matrix caused by node-specific target state estimates. 
We can see that H-ETP, while requiring minimal updating, is able to nearly obtain the performance of the centralized variant. 
}

\oldbluetext{blue}{MC does not demonstrate good performance because, due to the lack of feedback, it must conduct a lengthy exploration phase before a consensus is established. 
This phase is not concluded by the end of the simulation. }

Fig. \ref{fig:ecdf_end} goes on to show the performance of the proposed algorithms \emph{after} they have converged. 
This shows that the performance post convergence is improved over the early game performance. 
\oldbluetext{blue}{In particular, the performance gap between C-ETC and C-ETP is reduced.}
Note that H-ETP still obtains performance between these two. 
%
\begin{figure}
    \centering
    \includegraphics[scale=0.6]{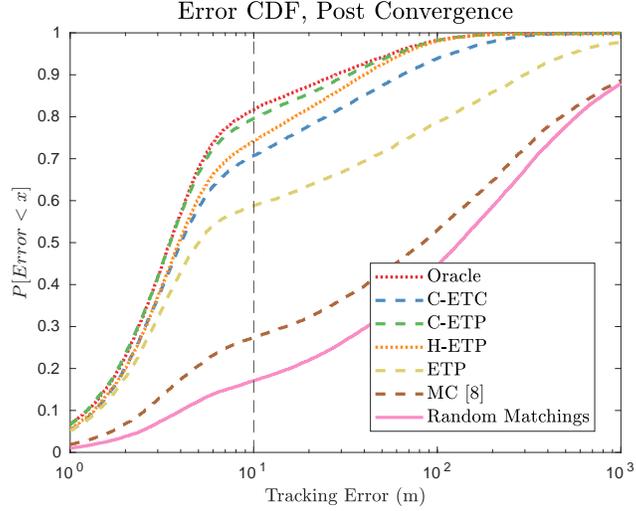}
    \caption{
    The error distribution for each algorithm after convergence. }
    \label{fig:ecdf_end}
\end{figure}

In Fig. \ref{fig:long_regret} we see that over a much longer time frame, the MC algorithm converges to an optimal matching. 
Due to the lack of feedback in MC, this process takes much longer (greater than $2500$ CPIs).

\begin{figure}
    \centering
    \includegraphics[scale=0.6]{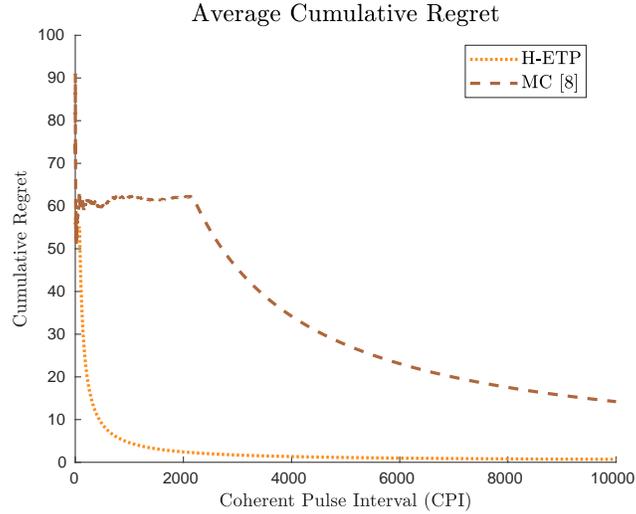}
    \caption{Cumulative regret for a much longer simulation. }
    \label{fig:long_regret}
\end{figure}

When the assumption on reward ordering is removed, we should expect that ETP will underperform, due to the lack of updates and estimation of rewards. 
H-ETP, however, would continue to receive updates as the environment evolves. 
Figure \ref{fig:no_ordering} demonstrates that when the ordering assumption is removed, ETP will slightly underperform. 
Importantly, the only reason that H-ETP continues to obtain near-optimal performance is the feedback rate. 
ETP is still able to learn throughout the game, but without continued feedback is unable to obtain optimal performance. 
This clearly demonstrates the tradeoff we study in this work: greater feedback leads to greater performance. 

\begin{figure}
    \centering
    \includegraphics[scale=0.6]{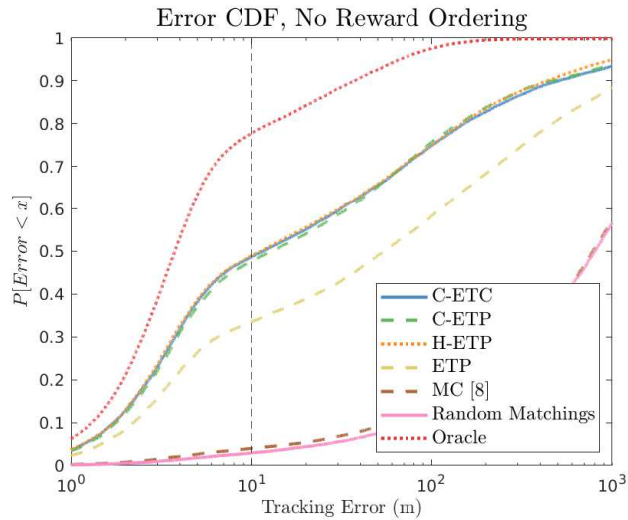}
    \caption{Error distributions without an assumption on reward ordering. We see that H-ETP and ETP both underperform due to the structure of $W^{\Gamma}$. }
    \label{fig:no_ordering}
\end{figure}
%


\section{Conclusions and Future Work}
In this work we have examined a hybrid cognitive radar network which has both cognitive nodes and a central coordinator. 
The network attempts to optimize channel selections for each radar over time in order to optimize radar tracking performance of a single target. 
This problem requires online machine learning techniques due to the desire for low convergence times and applicability across environment instances.

The network uses the observed \texttt{SINR} in each channel to inform rewards for the learning algorithms. 
Due to this, the nodes in the network must simultaneously learn both the interference behavior in each channel and estimate the target position, since both of these impact the \texttt{SINR} observed for transmitting in a given channel. 
This results in a coupled structure, which we use to form \emph{estimates} of the \texttt{SINR} at each node. 
This is useful because use of this estimate can reduce the amount of feedback \oldbluetext{blue}{needed from a central coordinator.}

We proposed and examined several online machine learning algorithms that this network can use to effectively balance this coupled learning problem. 
While previous work has investigated the general CRN learning problem, this work investigated how a network could converge faster, and thus experience better performance for a greater part of the game, through the use of limited feedback from the central coordinator.

We show that algorithms using this estimated \texttt{SINR} can perform almost as well as algorithms receiving full feedback, while algorithms with no feedback under-perform due to extended convergence times. 
This represents the trade-off between feedback and performance: generally, more feedback yields greater performance. 
However, the performance exhibits diminishing returns as feedback increases. 
Importantly, shortened convergence times only provide benefit under finite time horizons, since each technique will ultimately converge to the optimal solution.

Our simulations show that out of the algorithms studied, the centralized ETP algorithm provides the best performance at the highest feedback cost. 
The C-ETC algorithm, which does not adapt to the dynamic environment, is unable to match the performance of algorithms which do adapt to the environment. 
This performance gap is emphasized when analyzing the performance post convergence, once the algorithm attempts to exploit the information it has learned. 
This is due to the fact that the ETP algorithm is able to utilize information about the target to predict future range and \texttt{SINR}.

Our results show that H-ETP demonstrates the best trade-off between performance and feedback. 
For a moderate level of feedback (less than two floating-point values per CPI), H-ETP exhibits tracking performance almost as good as the best-performing algorithm C-ETP, which uses over 30 floating-point values per CPI.

In addition, these centralized algorithms show convergence times on the order of 100 CPIs, while prior work \cite{howard2021_multiplayerconf} required up to 2500 CPIs to meet the same target tracking performance without feedback. 
This improved convergence time translates directly to improved performance (i.e., track error in early CPIs) in a dynamic environment, where new optimal matchings may need to be identified rapidly. 
Further, improved convergence times cause less cumulative track error, since the target is accurately tracked much sooner. 
These improvements come at a cost of communication between the nodes and CC, as well as an implicit reliance on the CC. 
We demonstrated that there is a trade-off between the amount of feedback and the performance of the network, and that a minimal amount of feedback is sufficient to cause dramatic gains over the prior art.

While this work focused on a single-target scenario, a multi-target extension is straightforward, as the \texttt{SINR} will still be dependent on target location and environmental interference. 
However, as the number of targets increases, the amount of target parameter estimates sent from each node to the CC will also increase. 
For this reason, future work could examine the impact of CC-level node selection for target state updating. 
Further, the multi-target problem could present opportunities for each node to conduct both active radar tracking and passive observation. 
Future work could also examine the possibilities of CRN mode control.

\begin{appendix}
    \begin{proof}[Proof of Lemma \ref{lem:equivilence}]
        For an element $x_{m,n}\in\gamma(k)$, $y_{m,n}\in W^{\Gamma}(k)$ can be written as Eq. (\ref{eq:pf1}). 
        \begin{align}
        \label{eq:pf1}
            y_{m,n} &= \left(\frac{1}{\hat{r}_m^4(k)}\right) \Gamma_{m,n} \nonumber \\
            &= \left(\frac{1}{\hat{r}_m^4(k)}\right) \frac{\gamma_{m,n}(k)}{\hat{P}_m(k)} \nonumber \\
            &= \left(\frac{1}{\hat{r}_m^4(k)}\right) \gamma_{m,n}(k) \frac{(4\pi)^3\hat{r}_m^4(k)}{P_xG^2\lambda^2} \nonumber \\
            &= \frac{(4\pi)^3}{P_xG^2\lambda^2} \gamma_{m,n}(k) 
        \end{align}
        Now, since each element of $W^\Gamma$ is equal to a constant times $\gamma(k)$, the optimal matching function will not change. 
        \begin{equation}
            \label{eq:pf2}
            \max_{\pi\in\gamma(k)}U(\pi) = \max_{\pi\in W^{\Gamma}(k)}U(\pi)
        \end{equation}
    \end{proof}
\end{appendix}


\bibliographystyle{IEEEtran}
\bibliography{bibli}

\end{document}